\documentclass[journal]{IEEEtran}

\usepackage{amsmath}
\usepackage{amssymb}
\usepackage{amsthm}
\usepackage{cite}
\usepackage{graphicx}
\usepackage{booktabs}
\usepackage{array}
\usepackage[table]{xcolor}
\usepackage{pifont}
\usepackage{colortbl}
\usepackage{multirow}
\usepackage{tcolorbox}
\usepackage{xurl}
\usepackage[colorlinks,
  linkcolor={green!80!black},
  citecolor={red!70!black},
  urlcolor={blue!70!black}]{hyperref}

\tcbset{boxstyle/.style={boxrule=0.8pt,arc=1pt,left=4pt,right=4pt,top=2pt,bottom=2pt,before skip=4pt,after skip=4pt}}
\newtcolorbox[auto counter]{keyfinding}{boxstyle,colframe=blue!55!black,colback=blue!4!white,before upper={{\normalfont\bfseries\color{blue!55!black}Key Finding \thetcbcounter.}\hspace{0.5em}}}
\newtcolorbox[auto counter]{researchgap}{boxstyle,colframe=orange!75!black,colback=orange!5!white,before upper={{\normalfont\bfseries\color{orange!75!black}Research Gap \thetcbcounter.}\hspace{0.5em}}}

\newcommand{\cmark}{\textcolor{green!55!black}{\ding{51}}}
\newcommand{\xmark}{\textcolor{red!75!black}{\ding{55}}}
\definecolor{rowgray}{gray}{0.92}

\newtheorem{theorem}{Theorem}

\makeatletter
\def\paragraph{\@startsection{paragraph}{4}{2\parindent}%
  {-0ex plus -0.1ex minus -0.1ex}{0ex}{\normalfont\normalsize\bfseries}}
\makeatother

\begin{document}
\bstctlcite{TIFSAuthorListControl}

\title{The Safeguard Worked. Is the LLM System Safer?}

\author{Pingyu~Wu$^{1,2}$, Weiming~Zhang$^{1,*}$, and Nenghai~Yu$^{1}$\\[4pt]
{\normalsize $^{1}$University of Science and Technology of China, Hefei, China\\
$^{2}$Hefei AiDA Lab, Hefei, China\\
wupingyu@mail.ustc.edu.cn, zhangwm@ustc.edu.cn, ynh@ustc.edu.cn\\
$^{*}$Corresponding author}%
\thanks{This work has been submitted to the IEEE for possible publication. Copyright may be transferred without notice, after which this version may no longer be accessible.}}

\maketitle

\begin{abstract}
Safeguards for large language model services are typically evaluated by
their attack-blocking rates on test suites. Deployment decisions also depend on how much harmful assistance attackers can still obtain from the system, which these rates do not measure. To determine what published evaluations establish about
deployment safety, we develop a framework that translates reported results
into the minimum and maximum harmful assistance that can remain under a specified outcome and threat model. For safeguards that block attacks through a
check, we show that the most an attacker can obtain depends on three properties:
whether attack paths pass through the check, its miss rate against every
attacker in the threat model, and how much assistance remains after a block. We apply this framework to 85 safeguard evaluations from 52 papers. Of these, only five argue that every
attack path passes through the check, and four provide evidence on what
remains after a block. These evidence gaps have practical consequences: in NeMo Guardrails' default streaming configuration, the check blocked 94\%
of password leaks in our test, yet the client received the password in every leak, blocked or not.
We identify which of the three properties each safeguard category's
evaluations omit, and what to measure for each of four ways
attacks bypass a check.  Deployments that require these measurements
before launch can choose and configure safeguards that limit what
attackers obtain from the deployed system.

\end{abstract}


\section{Introduction}
\label{sec:introduction}

Production large language model (LLM) services deploy safeguards because their capabilities can be
misused~\cite{solaiman2019release,li2024wmdp,DBLP:conf/uss/ZhanCSS25,OpenAI2026GPT56SystemCard}.
Common safeguards include refusal behavior in the
model~\cite{bai2022constitutional,dai2024saferlhf,arditi2024refusal,zou2024circuitbreakers},
runtime classifiers that check inputs and
outputs~\cite{inan2023llamaguard,han2024wildguard,Cunningham2026ConstitutionalClassifiersPP},
account
monitoring~\cite{Anthropic2025MisuseReport,OpenAI2025DisruptingMaliciousUses},
access controls, tool policies, and
permissions~\cite{Debenedetti2025CaMeL,costa2025fides,shi2025progent}, and
execution containment~\cite{wu2025isolategpt,Anthropic2026Containment}.

Papers that evaluate safeguards report refusal and attack success
rates~\cite{mazeika2024harmbench,chao2024jailbreakbench}, policy
violations~\cite{Cunningham2026ConstitutionalClassifiersPP}, and utility on
benign tasks~\cite{rottger2024xstest,debenedetti2024agentdojo}.  These
numbers show how a safeguard performed against the requests and attacks in
the test, but real-world risk depends on the threat model and on how
attackers behave~\cite{peppin2025reality,lukosiute2025cyberrisk}.  A
deployment claim about all attackers in a declared threat
model concerns a different quantity from these test numbers: the largest harmful assistance any such
attacker can obtain, including through adaptation or bypasses of the
safeguards.  One successful attack shows that some harmful assistance remains, and tests quantify performance and improvements under the conditions they evaluate.  Extending test results to every attacker in the threat model requires additional evidence or assumptions about strategies the tests did not cover.  Attacks on adversarial-example defenses
showed that reported test results did not carry over to adaptive
attackers~\cite{CarliniWagner2017NotEasilyDetected,Tramer2020AdaptiveAttacks}.
The same held for adaptive attacks on 12 jailbreak and prompt-injection defenses:
they succeeded more than 90\% of the time against most of the 12, whereas
most of the original experiments had reported rates near
zero~\cite{Nasr2026AttackerMovesSecond}.

This paper asks
what a reported result, such as a high block rate, establishes about the largest harmful assistance an attacker in the threat
model can obtain.  Reported results range from one successful attack, through rates
measured on a fixed set of tests, to reachability analyses and proofs.  For
each kind, we derive the strongest conclusions it supports about how much
harmful assistance an attacker in the threat model can obtain, at least and
at most.

\begin{figure}[t]
  \centering
  \includegraphics[width=\columnwidth]{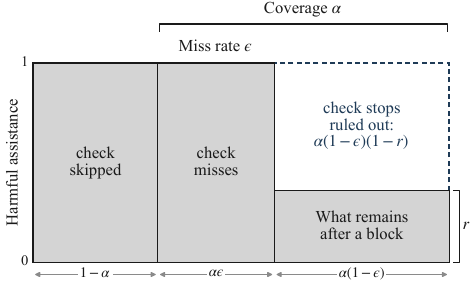}
  \caption{Harmful assistance obtainable in the three cases of a check.}
  \label{fig:gates}
  \par\vspace{2pt}
  \parbox{\columnwidth}{\footnotesize\raggedright Width is probability.
  The shaded area, $1-\alpha(1-\epsilon)(1-r)$, is the most an attacker can obtain.}
\end{figure}

For a safeguard that blocks attacks
with a check, the most an attacker can obtain depends on three properties, each guaranteed against every attacker in the threat model: the probability
that the check inspects an attack (\emph{coverage}), the probability that it
misses an attack it inspects (\emph{miss rate}), and how much the attacker
can still obtain after a block (Figure~\ref{fig:gates},
Theorem~\ref{thm:local-contract-lift}).  What remains after
a block includes output already released and further attempts.

These three properties are the key measurements for a check.  A
lower miss rate reduces what an attacker can obtain only for attacks the
check inspects, and only when a block keeps part of the output from the
attacker.  Coverage and what remains after a block depend on which paths the deployment sends through the check and on what it has released when the check acts, so evidence on them has to come from the deployment, and tests of the check alone cannot supply it.

We therefore ask whether published evaluations of LLM safeguards provide
evidence on these three properties
(Section~\ref{sec:theory-guided-systematization}).  From a random sample of the 1{,}906 screened papers on LLM attacks, defenses, and benchmarks, we coded 85 safeguard
evaluations in 52 papers.  These cover
six safeguard categories: training-time alignment, capability removal, input
and output classifiers, inference-time defenses, monitoring across turns, and system and agent
architecture.

In 67 of the 85 evaluations, an attack within the threat model obtains harmful assistance, including evaluations that show the safeguard improves on its baseline.  Only 5 argue full
coverage and 4 provide evidence on what remains after a block; the
2 that guarantee a miss rate against every attacker in the threat model both come from Fides~\cite{costa2025fides}, an information-flow control system for LLM agents.

Only Fides's two formal evaluations establish all three properties, so no
harmful assistance remains for the two outcomes they guarantee
(Section~\ref{subsec:upper-endpoint-results}).

Leaving these properties unreported has consequences in deployment.  NeMo Guardrails is an
open-source guardrail.  In our test of its default streaming
configuration, the output check blocked 94\% of password leaks, yet the
client received the password in every leak, blocked or not
(Section~\ref{subsec:release-measurement}).  A block rate does not show
how much the client has received.

Our contributions are:

\begin{itemize}
  \item \textbf{Key measurements for a check.}  For each kind of reported
    result, we derive the strongest conclusions it supports about how much
    harmful assistance an attacker can obtain, at least and at most.  For a
    check, the most an attacker can obtain is set by coverage, the
    miss rate against every attacker in the threat model, and what remains
    after a block (Section~\ref{sec:deployment-theory}).

  \item \textbf{Measurement of published evaluations.}  Only 2 of 85
    safeguard evaluations from 52 papers establish all three properties.
    We identify the measurements missing from each safeguard category's
    evaluations (Section~\ref{sec:deployment}).

  \item \textbf{Implications for deployment.}  For four ways that published
    work shows attacks bypass a check, we give the measurement that shows
    whether a deployment prevents each bypass, and the deployment change
    that targets it (Section~\ref{subsec:governance}).
\end{itemize}

\section{Related Work}
\label{sec:prior-systematizations}

Attacks not included in the original evaluations have shown that
safeguards reported as effective still let attackers obtain harmful
assistance: adaptive attacks on jailbreak and prompt-injection
defenses~\cite{Nasr2026AttackerMovesSecond}, further fine-tuning of
open-weight safeguards claimed to be durable~\cite{qi2025durability}, and
questions that each look benign yet together extract dangerous
knowledge~\cite{Glukhov2025BreachLeaks}.  Each result shows attacks that a safeguard fails to stop, but not what
remains when a check has zero miss rate.  We prove that even a check with zero miss rate can leave the attacker all
the harmful assistance when attacks avoid the check entirely or the check acts after release (Theorem~\ref{thm:local-contract-lift}).

Complete mediation requires every access to pass through a
check~\cite{SaltzerSchroeder1975}, and Schneider characterized the security
policies that mechanisms monitoring system execution can
enforce~\cite{Schneider2000Enforceable}.  Both presuppose that every access passes through a check that enforces the policy.
Theorem~\ref{thm:local-contract-lift} gives the most harmful assistance
that can remain when the check covers only some paths to the outcome,
misses some of the attacks it inspects, or leaves part of the output
obtainable after a block.

Safety cases and evaluation principles link safeguard evidence to
deployment decisions through red-team estimates entered into an uplift
model~\cite{Clymer2025SafetyCase} or a developer's justification that the
safeguards meet stated requirements~\cite{AISISafeguardsTeam2025Principles}.
The decision is then only as well supported as these estimates and justifications.  For each kind of reported result, we derive the
strongest conclusion the result itself supports under its declared outcome
and threat model, and on 85 published evaluations these conclusions show
that harmful assistance remains in 67 and that none remains for the two
outcomes Fides's formal configuration guarantees.

\section{Problem Formulation and Threat Model}
\label{sec:definitions-scope}

A reported safeguard result supports a deployment decision only under the
outcome, threat model, and operating conditions it covers.  We therefore fix
these conditions before defining the harmful assistance that remains.

\subsection{Fixing the Evaluation Setting}

The \emph{evaluation setting} consists of the six conditions under which a
deployment decision is made:
\begin{IEEEdescription}[\IEEEsetlabelwidth{$\Sigma$}]
  \item[$S$] \emph{Service}: the focal LLM service without the evaluated
    safeguard intervention.
  \item[$D$] \emph{Intervention}: the evaluated intervention, deployed in
    the guarded service $S[D]$ with all other parts of $S$ unchanged.
  \item[$\Sigma$] \emph{Threat model}: a declared nonempty set of attacker
    strategies with limits on resources and interaction.
  \item[$\mathcal L$] \emph{Constraints}: the constraints on basic rights
    and legitimate use that every admissible deployment must satisfy,
    including a utility requirement on benign use, as set by the
    application and its governing institutions.
  \item[$Z$] \emph{Outcome scale}: a scale from 0 to 1 that scores the
    harmful assistance the service provides toward the adverse outcome
    through its outputs, actions, and state transitions.  The scale measures
    assistance, not the harm the attacker ultimately causes with it;
    a score of 0 means the service provides nothing toward that outcome.
  \item[$\mathcal E$] \emph{Task, environment, and horizon}: the fixed
    task, target, operating environment, benign reference conditions, and
    evaluation horizon, including allowed interactions and retries, plus
    any persistence or recovery window relevant to the outcome.
\end{IEEEdescription}

Changing one condition can change the measured result sharply.  With all
else fixed, Jain et al.\ measured attack success of 0 to 1\% against a
first-round attacker and 7.9 to 16.8\% against one that adapts to defender
feedback over 15 rounds~\cite{Jain2026AdaptiveAdversaries}.  Changing the scored outcome also changes the result: FragFuse bypassed access control 86.3\% of the time but completed
the harmful task end to end 41.1\% of the
time~\cite{Rao2026FragFuse}.

A claim about how much harmful assistance a guarded service provides is
made within one such setting; the claim's scope is the set of attacker
strategies and operating conditions for which it asserts its conclusion.
Because the results of Section~\ref{sec:deployment-theory} take the outcome
and threat model from the setting, they apply equally to an action the
service must not take and to knowledge an attacker must not gain.

\subsection{Threat Model}
The deployment commits to its configuration, and the attacker then chooses
any strategy in the threat model, within its limits on resources and
interaction.  The attacker's goal is harmful assistance on the outcome scale.
A strategy may adapt to what the service releases, retry within the
evaluation horizon, and randomize.  For a released artifact such as model weights, the attacker keeps the artifact and may modify it in any way within the declared tampering budget.

\subsection{Residual Assistance and Tolerance}
\label{subsec:deployment-game}

An interaction is a complete run of an attacker's strategy against the
committed configuration over the evaluation horizon, including every
attempt the strategy makes.  Each interaction receives a score $v\in[0,1]$:
the most harmful assistance obtainable from it.  This score counts outputs already released, actions the deployment permits, and accumulated
state.

A strategy $\sigma\in\Sigma$ induces a distribution over interactions, with
expectation $\mathbb E_\sigma$ and probability $P_\sigma$.  The distributions
over interactions that the attacker can induce, with or without randomizing,
are the distributions the attacker can reach.
The \emph{harmful assistance that remains} is the largest expected score on
the outcome scale that an attacker in the threat model can obtain from the guarded service,
\begin{equation}
  Z(S[D])=\sup_{\sigma\in\Sigma}\mathbb E_\sigma[v],
  \label{eq:residual-value}
\end{equation}
and an attacker who randomizes over strategies obtains no more.
A deployment decision depends on the harmful assistance that remains.
The difference from the unguarded service, ${Z(S)-Z(S[D])}$,
measures what the intervention
removes~\cite{Mahato2026SafeguardConditionedUplift}.
Because removed and remaining add up to the unguarded amount, a reported reduction alone does not show how much remains.

A deployment declares a tolerance $\tau$ and requires $Z(S[D])\leq\tau$
while satisfying the constraints $\mathcal L$, which limit how much it can
refuse legitimate users to meet the tolerance.  A \emph{certificate} is
evidence that establishes an upper bound on the harmful assistance that
remains.

\section{Theoretical Analysis}
\label{sec:deployment-theory}

For a fixed evaluation setting, reported evidence can show that the
harmful assistance that remains is at least some amount (a lower bound) or
at most some amount (an upper bound).  A bound is \emph{tight}
when a deployment consistent with that evidence attains it, so the evidence
alone supports no stronger bound.  A stronger conclusion requires evidence
that rules out every deployment that attains the bound.

\subsection{Lower Bounds from Executed Attacks}
\label{subsec:lower-certificates}

Every strategy $\sigma$ in the threat model shows that an attacker can
obtain at least its expected score.
\begin{equation}
  Z(S[D])\geq\mathbb E_\sigma[v].
  \label{eq:attack-lower-bound}
\end{equation}
This bound requires less evidence than any other in this paper: a single
attack within the threat model, run against the deployed configuration.
A reported attack success rate scored on the declared outcome estimates
this expected score.  One scored success establishes that some harmful assistance remains, which already exceeds a zero tolerance.  Only a comparison with a tolerance $\tau>0$ requires the size of
the rate, using a one-sided confidence limit computed from the reported
counts.  For safeguards whose pre-release evidence an attacker can reproduce, Wu et al.~\cite{Wu2026CopyableContext} show that measurements on benign users already establish such a lower bound, without an executed attack.

\subsection{Upper Bounds from Reachability and Blocking}
\label{subsec:upper-certificates}

To show that an attacker can obtain at most some amount, evidence must
cover every distribution the attacker can reach.
Two methods achieve this.  One uses a set that contains all distributions
the attacker can reach.  The other combines how often the safeguard blocks
the attack with what remains obtainable after a block.

\subsubsection{Reachability}
\label{subsec:direct-reachable-upper}

Suppose evidence establishes a set $\mathcal M$ of distributions that
contains every distribution the attacker can reach.  Then
$Z(S[D])\leq\sup_{\mu\in\mathcal M}\mathbb E_\mu[v]$, with equality when
$\mathcal M$ contains no other distribution.

This method applies to artifacts that no check inspects after release,
such as released weights.  The analysis must cover every modification
within the declared tampering budget and show that each resulting system
provides at most the stated amount.  However many modifications are tested, they only sample the threat model and cannot rule out that an untested modification obtains full harmful assistance~\cite{Athalye2018ObfuscatedGradients,Tramer2020AdaptiveAttacks}.

\subsubsection{Blocking Events}
\label{subsec:success-region}

Suppose that under every strategy in the threat model the safeguard
applies its blocking response with probability at least $\lambda$, and that
at most an amount $r$ remains obtainable from every attacker-reachable
attack it blocks (i.e., $v\leq r$ on every such interaction).  Then
\begin{equation}
  Z(S[D])\leq1-\lambda(1-r),
  \label{eq:blocking-bound}
\end{equation}
and no smaller bound follows from these two quantities.  An attacker
obtains exactly $1-\lambda(1-r)$ when the safeguard blocks with probability
exactly $\lambda$, a blocked attack scores $r$, and every other attack
scores one.  Applying the bound requires two
pieces of evidence for every attacker: that the safeguard blocks with
probability at least $\lambda$, and that at most $r$ remains obtainable
after a block.  The evidence on $r$ must cover all remaining state, interfaces, and attempts.

\subsection{The Three Properties for a Check}
\label{subsec:closed-mediation}

A check's miss rate, such as a classifier's, limits what an attacker can
obtain only together with two other properties of the check.
\emph{Mediation} means that attack paths pass through a check that can stop
them.  Let $A$ be the event that the attack path enters the
part of the system the check inspects, and $F$ the event that the check
misses the attack there.  The \emph{coverage} $\alpha$ is the least probability of $A$ over the
threat model, and the check has \emph{miss rate} $\epsilon$ if it misses at
most a fraction $\epsilon$ of the attacks it inspects:
\begin{equation}
  \begin{gathered}
  \alpha=\inf_{\sigma\in\Sigma}P_\sigma(A),\\
  P_\sigma(F\cap A)\leq\epsilon P_\sigma(A)\ \ \text{for every }\sigma\in\Sigma.
  \end{gathered}
  \label{eq:coverage-miss-rate}
\end{equation}
We take the least probability because the attacker chooses the path: one strategy
that reaches the outcome without entering the inspected part gives coverage
zero.  A miss rate measured on a test distribution does not establish $\epsilon$, which must hold under every strategy in the threat model.  With
coverage $\alpha$ and miss rate $\epsilon$, the check stops the attack with
probability at least $\alpha(1-\epsilon)$ under every strategy; this is the
blocking probability $\lambda$ in Equation~\ref{eq:blocking-bound}.  The
third property, what remains obtainable after a block, is the amount $r$ in
that equation: the most harmful assistance still obtainable from an attack
that the check inspects and stops.

The complete-mediation principle~\cite{SaltzerSchroeder1975} requires every
access to pass through a check, which is full coverage.  A check on an LLM service, however, can leave some attack paths unchecked ($\alpha<1$), can miss attacks
it inspects ($\epsilon>0$), and can block only after part of a reply has
reached the client ($r>0$).
Theorem~\ref{thm:local-contract-lift} gives the most harmful assistance an
attacker can obtain under such a check.

\begin{theorem}[Tight mediation bound]
\label{thm:local-contract-lift}
Suppose that a check has coverage $\alpha$, misses at most a fraction
$\epsilon$ of the attacks it inspects under every strategy in the threat
model, and leaves at most $r$ obtainable from every attacker-reachable attack
that it inspects and stops.  Then
\begin{equation}
  Z(S[D])\leq1-\alpha(1-\epsilon)(1-r),
  \label{eq:closed-mediation-bound}
\end{equation}
and this bound is tight.  The bound is below one if and only if coverage
is positive, the miss rate is below one, and a block withholds something
($\alpha>0$, $\epsilon<1$, $r<1$).
\end{theorem}

Theorem~\ref{thm:local-contract-lift} states that the check reduces what
an attacker can obtain only through attacks it inspects and stops, and only by what a block
withholds, $1-r$.  The bound follows from the three cases in
Figure~\ref{fig:gates}.  An attack that avoids the check or that the check
misses can score at most one; an attack that the check stops scores at most
what remains after the block.  In the deployment in the figure, an attacker
obtains exactly $1-\alpha(1-\epsilon)(1-r)$: each case occurs with the probability shown and obtains the largest
score the case allows.  Evidence on the three properties shows that no harmful assistance remains exactly when coverage is full, the miss rate is zero, and nothing remains after a block.

Lowering the miss rate from $\epsilon$ to $\epsilon'$ reduces the most an
attacker can obtain by exactly $\alpha(\epsilon-\epsilon')(1-r)$, which is zero exactly when
coverage is zero ($\alpha=0$) or everything can remain after a block ($r=1$; Figure~\ref{fig:gates}).

\begin{keyfinding}
Even with zero miss rate, a check that acts after the reply is sent can
leave the attacker all the harmful assistance.
\end{keyfinding}

\subsection{Multiple Checks and Repeated Attempts}
\label{subsec:deployment-composition}

For a stack of checks, the miss rate that evidence can guarantee depends on
whether the layers miss the same attacks.  Suppose an attack path requires every layer to miss it in sequence.  If
each layer's miss rate holds after every history that reaches it, the
probability that all layers miss is at most the product of their miss
rates, with no independence assumption.  Rates measured on each layer separately do not justify this product, because a later layer's rate must hold on the attacks that the earlier layers have already missed.
With only such separate rates, as for input and output classifiers
evaluated separately, the evidence guarantees only the miss rate of the
strongest layer: layers that miss the same attacks do no better.

\begin{keyfinding}
Stacked checks tested separately guarantee only the miss rate of their
best layer.
\end{keyfinding}

A per-attempt test leaves out later attempts, which Theorem~\ref{thm:local-contract-lift} requires $r$ to cover.  When each attempt
passes through the check separately, later attempts can instead be counted
in the blocking probability by taking the block in
Equation~\ref{eq:blocking-bound} to be the event that the check blocks
every attempt.  Suppose the evaluation horizon allows $H$ attempts, coverage
$\alpha$ and miss rate $\epsilon$ hold for each attempt after any earlier history, and at
most an amount $r_0$ remains obtainable after all attempts are blocked.  The
check then blocks every attempt with probability at least
$[\alpha(1-\epsilon)]^H$, so
\begin{equation}
  Z(S[D])\leq1-[\alpha(1-\epsilon)]^H(1-r_0),
  \label{eq:retry-mediation-bound}
\end{equation}
which for $H=1$ is Theorem~\ref{thm:local-contract-lift}, with $r_0=r$.
With the per-attempt parameters and $r_0$ fixed, the most an attacker can
obtain grows with the number of allowed attempts.

\begin{table*}[t]
  \centering
  \footnotesize
  \caption{What reported evidence establishes about the harmful
  assistance that remains.}
  \label{tab:exchange-rate}
  \renewcommand{\arraystretch}{1.2}
  \setlength{\tabcolsep}{4pt}
  \setlength{\aboverulesep}{0pt}
  \setlength{\belowrulesep}{0pt}
  \setlength{\extrarowheight}{1pt}
  \rowcolors{2}{white}{rowgray}
  \begin{tabular}{>{\raggedright\arraybackslash}p{145pt}
    l c
    >{\raggedright\arraybackslash}p{105pt}
    >{\raggedright\arraybackslash}p{125pt}}
    \toprule\addlinespace[2pt]
    Reported quantity & Bound & \shortstack{Gives a\\bound}
      & Consequence & Reason \\
    \midrule
    Expected harmful assistance obtained by one attacker strategy in the threat model
      (Sec.~\ref{subsec:lower-certificates})
      & Lower & \cmark & at least this amount
      & the executed attack obtains it \\
    Enumerated attacks inside a declared budget, such as a benchmark suite
      (Sec.~\ref{subsec:direct-reachable-upper})
      & Upper & \xmark & cannot rule out full harmful assistance
      & an untested attack in the threat model can obtain more \\
    Reachability analysis bounding every reachable system
      (Sec.~\ref{subsec:direct-reachable-upper})
      & Upper & \cmark & at most this amount
      & the analysis covers every reachable system \\
    Miss rate alone, such as a classifier's, at any value including zero
      (Thm.~\ref{thm:local-contract-lift})
      & Upper & \xmark & cannot rule out full harmful assistance
      & the attack can avoid the check entirely, or everything can remain
        after a block \\
    Coverage ($\alpha$), miss rate ($\epsilon$), and what remains obtainable
      after a block ($r$) (Thm.~\ref{thm:local-contract-lift})
      & Upper & \cmark & at most $1-\alpha(1-\epsilon)(1-r)$
      & an attacker obtains exactly this (Fig.~\ref{fig:gates}) \\
    Miss rate of each layer, measured separately
      (Sec.~\ref{subsec:deployment-composition})
      & Upper & \cmark & stack miss rate at most the strongest layer's
      & layers that miss the same attacks miss exactly this often \\
    Miss rate of each layer, holding after every earlier history
      (Sec.~\ref{subsec:deployment-composition})
      & Upper & \cmark & stack miss rate at most the product of the layers'
        miss rates
      & independent layers miss exactly this often \\
    \bottomrule
  \end{tabular}
  \par\vspace{3pt}
  \parbox{0.97\textwidth}{\raggedright \cmark: the quantity establishes the
  bound shown on its own.  \xmark: it cannot.}
\end{table*}

\subsection{Matching Evidence to Bounds}
\label{subsec:exchange-rate}

Table~\ref{tab:exchange-rate} lists what each kind of evidence
establishes.  Lower bounds rest on evidence about one attacker;
upper bounds require evidence covering every attacker-reachable
distribution.

Multiple rows may apply to one deployment.  Within one evaluation setting,
the evidence then establishes the largest of their lower bounds and the
smallest of their upper bounds; a side with no bound stays open.  An upper bound that
covers only part of the threat model must first be extended to all of it.
If the largest lower bound exceeds the smallest upper bound, their
premises cannot all hold under one evaluation setting, and the conflict must
be resolved first.

\label{subsec:evidence-and-conclusion-states}
The evidence shows that the deployment meets its tolerance $\tau$ when
it shows that at most $\tau$ remains, and exceeds it when it shows that more
than $\tau$ remains; otherwise the question is unresolved.  At zero
tolerance, meeting it requires evidence that none remains.

\section{Measurement of Published Evaluations}
\label{sec:theory-guided-systematization}
\label{sec:deployment}

Section~\ref{sec:deployment-theory} shows which evidence can establish how
much harmful assistance an attacker can obtain at least or at most.
We check whether 85 published safeguard evaluations report this evidence: most show that
harmful assistance remains, and almost none report the evidence that could
show how much an attacker can obtain at most.

\subsection{Coding Method}
\label{subsec:claim-instance-anchor}
\label{subsec:bilateral-evidence-matrix}
\label{subsec:matrix-review-method}

The unit of analysis is a safeguard evaluation: a test, reported in a source, of whether a deployed safeguard prevents harmful assistance.  An evaluation consists of the deployed policy under test, its evaluation setting, and the
source's assertion.  The test must be a comparison (guarded versus unguarded, versus a
baseline defense, or before versus after) or a formal statement: a
theorem, an invariant, or an argument that the architecture cannot be
bypassed.
A paper contains several evaluations when it varies the
intervention, outcome, threat model, utility requirement, or evaluation horizon.
Measurements belong to one evaluation only when they support the same assertion
in the same evaluation setting.  Attack studies outside the source are coded
against the evaluation they test, so that later evidence tests the original assertion.

Each evaluation is assigned to one category by the safeguard mechanism it tests.
Training-time alignment changes the weights before deployment so that the
model avoids harmful outputs; capability removal aims for the model to lack
the harmful capability, or to keep lacking it after later modification;
input/output classifiers decide, separately from generation, whether to block
or flag one input or output; inference-time defenses change generation without
changing the weights; monitoring across turns decides from state kept across
turns, sessions, or accounts; system and agent architecture limits what
reaches the model or what its outputs can do.

We searched arXiv, the ACL Anthology, and dblp proceedings of major
security and machine-learning venues for work through August 2026 that names an LLM system together with a defense, an evaluation, or an adverse outcome such as jailbreak or prompt injection.
The search returned 19{,}276 records.  We first kept work peer-reviewed at a fixed list of these
venues or cited at least ten times per year.  Two screening models, Claude
Opus~5.5~\cite{Anthropic2026Fable} and GPT-5.6 Sol~\cite{OpenAI2026GPT56},
then judged titles and abstracts without seeing each other's labels.  A record
advanced only if both models included it, which left 1{,}906.  The
same models then read full texts drawn in random order from these 1{,}906, retained papers that test whether a safeguard prevents harmful assistance, and independently coded each evaluation.
Of the 197 papers read, 52 were kept, with 85 evaluations.  The four
shares reported in Section~\ref{subsec:evidence-overview} had converged: over
the last 50 papers read, each changed by less than 2 percentage points (1.1, 0.5, 1.3, and 1.0;
Figure~\ref{fig:convergence}).

Table~\ref{tab:field-schedule} lists the coded fields.  An attack that
succeeds shows that harmful assistance remains
(Section~\ref{subsec:lower-certificates}); the three property fields show
how much remains at most (Theorem~\ref{thm:local-contract-lift}) only when one
evaluation establishes all three.
The three properties apply to safeguards built into the model or its
generation as well: the guarded model serves as the check, and a refusal or an
answer lacking the removed capability counts as a block.
An evaluation establishes a field by a result, measurement, architectural
property, or trace, each in the required setting and for one attacker or
every attacker as the field requires, or by a
stated inference rule from such fields.

\begin{table}[!tbp]
  \centering
  \footnotesize
  \caption{Coded fields and their values.}
  \label{tab:field-schedule}
  \renewcommand{\arraystretch}{1.2}
  \setlength{\tabcolsep}{3pt}
  \setlength{\aboverulesep}{0pt}
  \setlength{\belowrulesep}{0pt}
  \setlength{\extrarowheight}{1pt}
  \rowcolors{2}{rowgray}{white}
  \begin{tabular}{>{\raggedright\arraybackslash}p{56pt}
    >{\raggedright\arraybackslash}p{160pt}}
    \toprule
    Field & Values \\
    \midrule
    \multicolumn{2}{l}{\emph{Every evaluation}} \\
    Category & one of the six categories of safeguard \\
    \midrule
    \multicolumn{2}{l}{\emph{Harmful assistance remains
      (Sec.~\ref{subsec:lower-certificates})}} \\
    Attack succeeds & yes if an executed attack within the evaluated threat
      model obtains harmful assistance scored on the evaluated outcome;
      else no \\
    \midrule
    \multicolumn{2}{l}{\emph{How much remains at most
      (Thm.~\ref{thm:local-contract-lift})}} \\
    Miss-rate evidence & guaranteed against every attacker in the threat model, refuted by an attack from the claimed threat model,
      measured with adaptive attacks, measured on a fixed suite, or none \\
    Coverage evidence & argued full, bypass shown, only that the check is
      invoked, or none \\
    Evidence after a block & yes if the source measures or proves what remains
      obtainable after a block, released output and retries included;
      else no \\
    \bottomrule
  \end{tabular}
\end{table}

We count a successful attack when its strategy is feasible under the fixed
deployment, lies within the threat model, and is scored on the evaluated
outcome as the source itself reports it.
A bypass count, or any quantity not scored on the evaluated outcome, can
record the guarded versus unguarded comparison but cannot show that harmful
assistance remains.

We assess each source's claim against the threat model it claims, which
can be broader than the one it evaluates.  The claim is refuted when an
executed attack from that threat model contradicts it.  An attack that
obtains some harmful assistance does not refute an improvement claim.  A refuting attack outside
the evaluated threat model shows nothing about how much harmful assistance
remains within it.

The codebook was calibrated on 10 randomly drawn papers and then fixed.
Between the two models, Cohen's $\kappa$ was 0.92 for the category, 0.82 for
whether an attack succeeds, 0.68 for the miss-rate evidence, 0.70 for the
coverage evidence, and 0.65 for the evidence after a block.
The authors resolved the 37 field disagreements against the codebook.

\subsection{Evidence by Category of Safeguard}
\label{subsec:evidence-overview}

Of the 76 evaluations in the five categories other than system and agent
architecture, 1 argues full coverage and none reports evidence after a block
(Table~\ref{tab:evidence-profile}), so for none of them would even a proved
miss rate rule out that an attacker obtains all the harmful assistance.  Coverage and what remains after a block need no formal proof: an evaluation can argue coverage from the paths to the outcome and measure what remains (Section~\ref{subsec:release-measurement}).

In 67 of the 85 evaluations, an executed
attack within the threat model succeeds, so harmful assistance remains.
On the miss rate, only 2 of the 85 evaluations, both from Fides's formal guarantees, establish a rate that holds against every attacker in the threat model.  Of the other 83, 52 use a fixed test suite, 16 use adaptive attacks built against the defense, 12 make a claim refuted by an attack within the threat model they claim, and 3 report no rate.  Five evaluations argue that every
path to the outcome passes through the check: Fides's two formal
evaluations, the two of FedSEA-LLaMA~\cite{DBLP:conf/aaai/ZhangZLZDTZ26}, and WMDP's
unlearning~\cite{li2024wmdp}.  Four provide evidence on what
remains obtainable after a block: CoreGuard, CaMeL, and Fides's two formal
evaluations.  Only Fides's two formal evaluations establish all three
properties.
With 95\% confidence intervals from a bootstrap that resamples papers, the
shares are 78.8\% [68.2\%, 88.1\%] with a successful attack, 2.4\% [0.0\%,
7.4\%] with a guaranteed miss rate, 5.9\% [0.0\%, 13.4\%] with full
coverage, and 4.7\% [0.0\%, 10.6\%] with evidence after a block
(Figure~\ref{fig:convergence}).

\begin{table}[t]
  \centering
  \footnotesize
  \caption{Evidence reported by the 85 safeguard evaluations, by category of
  safeguard.}
  \label{tab:evidence-profile}
  \renewcommand{\arraystretch}{1.2}
  \setlength{\tabcolsep}{2pt}
  \setlength{\aboverulesep}{0pt}
  \setlength{\belowrulesep}{0pt}
  \rowcolors{2}{white}{rowgray}
  \begin{tabular}{>{\raggedright\arraybackslash}p{80pt}rrrrr}
    \toprule\addlinespace[2pt]
    \scriptsize Category of safeguard & \scriptsize Total
      & \scriptsize\shortstack[r]{Attack\\succeeds}
      & \scriptsize\shortstack[r]{Guaranteed\\miss rate}
      & \scriptsize\shortstack[r]{Full\\coverage}
      & \scriptsize\shortstack[r]{Evidence\\after a\\block} \\
    \midrule
    Training-time alignment & 29 & 28 & 0 & 0 & 0 \\
    Capability removal & 9 & 8 & 0 & 1 & 0 \\
    Input/output classifiers & 16 & 5 & 0 & 0 & 0 \\
    Inference-time defenses & 21 & 19 & 0 & 0 & 0 \\
    Monitoring across turns & 1 & 1 & 0 & 0 & 0 \\
    System and agent architecture & 9 & 6 & 2 & 4 & 4 \\
    \midrule
    \rowcolor{white}
    All & 85 & 67 & 2 & 5 & 4 \\
    \bottomrule
  \end{tabular}
  \par\vspace{3pt}
  \parbox{\columnwidth}{\raggedright
  \emph{Attack succeeds}: an executed attack within the threat model obtains
  harmful assistance.  \emph{Guaranteed miss rate}: a miss rate that holds against
  every attacker in the threat model.  \emph{Full coverage}: an argument that
  every path to the outcome passes through the check.  \emph{Evidence after
  a block}: evidence on what remains obtainable after a block, measured or
  proved.}
\end{table}

\begin{figure}[t]
  \centering
  \includegraphics[width=\columnwidth]{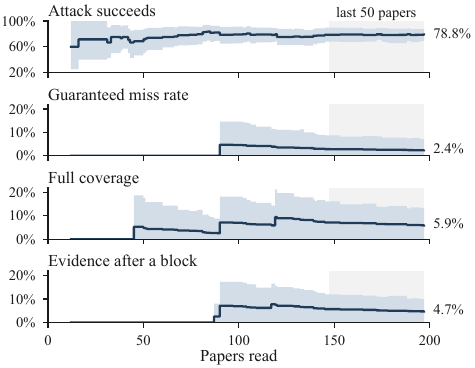}
  \caption{Shares of evaluations by papers read.}
  \label{fig:convergence}
  \par\vspace{2pt}
  \parbox{\columnwidth}{\footnotesize\raggedright Shaded bands are 95\%
  confidence intervals from a bootstrap that resamples papers. The gray
  area marks the last 50 papers read.}
\end{figure}

\subsection{Evidence That Harmful Assistance Remains}
\label{subsec:lower-endpoint-results}

An executed attack within the source's threat model establishes that
harmful assistance remains; the same evaluation can also support the
source's improvement claim.
RESTA reduces the harmful-response rate of the Chinese-Alpaca PEFT model on
Chinese CATQA from 55.67\% to 38.90\%~\cite{bhardwaj2024resta}.  The improvement is supported, and the residual 38.90\%
on the same scoring shows that harmful assistance remains.

An attack within a source's claimed threat model can also refute that
claim.
In the Circuit Breakers evaluation, PAIR, an attack in the source's own
suite, still obtains harmful output from the defended Mistral model on
23.3\% of HarmBench behaviors, showing that harmful assistance
remains~\cite{zou2024circuitbreakers}.  The source also claims robustness to
powerful unseen attacks.  A later adaptive reinforcement learning attack
that lies within that claimed threat model reaches 100\% attack success on
HarmBench~\cite{Nasr2026AttackerMovesSecond}, refuting the broader claim.

\subsection{Evidence on Coverage}
\label{subsec:coverage-results}

A successful attack can pass through the check and be missed, or follow a
path that the deployment never sends through the check.  The moderation system guarding
OpenAI's fine-tuning API blocks poisoned harmful-QA jobs for GPT-4o, misses a
code-backdoor job that teaches sleeper-agent behavior to GPT-4o mini at a 2\%
poisoning rate, and does not check fine-tuned GPT-3.5 Turbo and GPT-4
models at all~\cite{DBLP:conf/aaai/BowenMCKGP25}.
Fine-tuning either unchecked model therefore avoids the check entirely ($\alpha=0$).

For a stateful detector, coverage also depends on whether it keeps state
globally or per account, as two detectors built for image classifiers
show.  Blacklight keeps a
global store~\cite{li2022blacklight}, so all queries enter the collision
check and coverage is full.  Against a detector that keeps state per account, an attacker who
opens a new account resets that state.  The attack's earlier queries then
lie outside what the check inspects, and coverage is no longer
full~\cite{feng2023stateful}.

Arguing full coverage requires mapping the paths to the outcome, including paths across accounts or sessions, and no classifier or monitoring evaluation in Table~\ref{tab:evidence-profile} does so.  For training-time alignment and capability removal, attacks
through fine-tuning can bypass the check in deployments that release weights
or offer fine-tuning, so reachability analysis over the declared fine-tuning
budget, including a second fine-tuning round, can show how much harmful
assistance remains at most (Section~\ref{subsec:direct-reachable-upper}).

\subsection{Evidence on the Miss Rate}
\label{subsec:family-profile}

Outside system and agent architecture, the evaluations measure the miss rate
only on the attacks they ran.
Adaptive evaluations test defense-aware strategies and extend the
evidence beyond fixed suites, but they still test only some attacks in the threat model.  DeepAlign fine-tunes the model with losses on its hidden states
and limits an adaptive attack built against it to between 0\% and 9.8\%
success across five models~\cite{DBLP:conf/ndss/ZhangLZMC26}.
A miss rate shown to hold after every history the threat model admits
also holds for adaptive attackers.

In 28 of the 29 training-time alignment evaluations, an attack succeeds.  Pruning the unsafe
parameters of Mistral-7B lowers GCG success from 53.2\% to 17.5\% and PAIR
success from 54.0\% to 21.8\%~\cite{DBLP:conf/acl/SiLBZ26}.  The improvement
is supported, and the attacks that still succeed show that harmful assistance
remains.  Across
training-time alignment, a miss rate that holds against every attacker in
the threat model is still missing.

For capability removal and other safeguards that claim to survive
fine-tuning, the tested fine-tuning attacks cover only part of the tampering budget
(Section~\ref{subsec:direct-reachable-upper}).  Paladin implants
trigger tags so that phishing output carries a detection signal, and it
claims the tags persist even after the model is
modified~\cite{pang2026paladin}.  Malicious LoRA fine-tuning on 250 untagged
phishing samples lowers detection accuracy by more than 30\% for the
explicit-trigger, explicit-tag configuration, refuting that claim for
this configuration and budget.  Panacea perturbs the weights after fine-tuning and lowers the
average harmful score from 40.3\% to 19.1\% over harmful-data ratios up to
0.2, yet at a higher harmful-data ratio of 0.5 the harmful score still reaches 65.5\%~\cite{DBLP:conf/nips/WangHSYLLTHT25}.

Only 2 of the 16 input and output classifier evaluations measure the miss
rate against adaptive attacks.  One of them,
SelfDefend, passes each query to a shadow LLM
alongside the target and blocks the answer when the shadow flags it.  Attacks aimed at both models still succeed, with average success of 0.193 under its direct-detection
prompt and 0.023 under its intent-detection
prompt~\cite{DBLP:conf/uss/WangWJL00L0LR25}.  The 16
classifier evaluations contain no full-coverage argument and no evidence on
what remains after a block (Table~\ref{tab:evidence-profile}).  Supplying
both, together with a miss rate that holds against every attacker in the threat
model, would show how much harmful assistance remains at most.

A check that inspects every message can still miss information it cannot
decode.  Motwani et al.\ instruct a financial-advice agent that an overseer
monitors all its conversations and that insider advice is
prohibited~\cite{DBLP:journals/corr/abs-2402-07510}.  Once the agent and its
user share a codebook the overseer lacks, the agent conceals the advice in
innocuous prose.

In system and agent architecture, the design can establish the miss rate
and also determines coverage by controlling which paths reach the
protected outcome.
CoreGuard keeps an authorization inside a TEE, and all input- and
output-processing layers of an edge-deployed model depend on
it~\cite{DBLP:conf/nips/LiLZXSZJPZZY25}.  Without the authorization, the locked model scores 0\% accuracy: measured by accuracy, nothing remains after a block.
An adaptive attack that trains a substitute for the TEE-authorized
feed-forward layer from its inputs and outputs reaches only about the accuracy of a black-box attacker.  A miss rate that holds against every attacker in the threat
model is still missing.

\subsection{Evidence on What Remains After a Block}
\label{subsec:continuation-results}

A block limits the attacker only by what it withholds.
Ablove et al.\ study Chinese LLM services that apply output moderation
to politically sensitive content and measure what a block withholds.  Baidu-Chat sends the client a blocked reply in full and truncates only
the page rendering.  DeepSeek and Qwen deliver blocked replies that are
on average close in length to complete
ones~\cite{DBLP:conf/ndss/AbloveCQE26}, so the client receives most of a blocked reply.

Control flow after a check can also release information.  In CaMeL's
NORMAL mode, an exception that depends on private data changes whether a
later observable tool executes, revealing one private
bit~\cite{Debenedetti2025CaMeL}.

Chen et al.\ report defense success up to 93\% for their sequential
monitor on cumulative decomposition attacks~\cite{chen2026monitoring}.  This rate measures the monitor's miss rate only on the evaluated sequences.  A flag halts further answers in the
current conversation.  Harmfulness Probing similarly stops generation and
replaces the flagged answer with a refusal~\cite{zou2024circuitbreakers}.
After a block, the attacker retains the answers delivered before the
flag and can retry within the horizon; the reported results do not measure how much the attacker obtains this way.

Suppose the sequential monitor has full coverage and its best reported
result, a miss rate of at most 0.07, holds against every attacker in the threat model.  The reported results establish neither assumption.
Theorem~\ref{thm:local-contract-lift} then shows that an attacker obtains at
most 0.07 plus 0.93 times what remains obtainable after a block.  Evidence that at most
about 0.032 remains after a block, including retries in the horizon, would
show that an attacker obtains at most 0.1, meeting a tolerance of 0.1.

Alternatively, count one retry as a second attempt.  Suppose full
coverage and the sequential monitor's miss rate of at most 0.07 hold after every earlier history.
Even if nothing remains after both blocks, one retry raises the most this
evidence allows an attacker to obtain from 0.07 to $1-0.93^{2}\approx0.135$
(Equation~\ref{eq:retry-mediation-bound}), so this evidence alone no longer shows
that the 0.1 tolerance is met.  Under the same assumptions, a per-attempt
miss rate of at most about 0.051 would show that an attacker obtains at most
0.1.  Evidence that limits reachable retries or captures their dependence
could also show that an attacker obtains less than 0.135.

\begin{keyfinding}
Almost all of the 85 evaluations count a blocked attack as a defense success without reporting what the attacker still obtained.
\end{keyfinding}

For the five categories other than system and agent architecture, measuring assistance on the
evaluated outcome after a block or refusal would show what a block
withholds.

\subsection{Evidence on Multiple Checks and Repeated Attempts}
\label{subsec:failure-results}

When each layer's miss rate is established separately, a stack's miss rate
is at most that of the strongest layer
(Section~\ref{subsec:deployment-composition}).  Evidence about joint failures can
show that the stack misses less often within the interactions that evidence covers.

Among inference-time defenses, SmoothLLM provides evidence about joint failures within a single invocation.  The
original algorithm returns one of the sampled responses whose
jailbreak label agrees with its majority vote~\cite{robey2024smoothllmv4}.  When most
responses are labeled not jailbroken, the returned one scores zero.  For a given prompt, the perturbations are independent and identically
distributed, so each response is labeled jailbroken with the same probability
regardless of earlier labels.  The chance that most responses are labeled
jailbroken, a miss, is then a binomial tail.  This calculation covers one invocation; attacker-selected
retries across invocations need evidence on their histories.

\begin{table}[!b]
  \centering
  \footnotesize
  \caption{Password delivery under three release orders.}
  \label{tab:nemo-release}
  \renewcommand{\arraystretch}{1.2}
  \setlength{\tabcolsep}{2pt}
  \setlength{\aboverulesep}{0pt}
  \setlength{\belowrulesep}{0pt}
  \rowcolors{2}{white}{rowgray}
  \begin{tabular}{>{\raggedright\arraybackslash}p{84pt}rrrr}
    \toprule\addlinespace[2pt]
    \scriptsize Release order
      & \scriptsize\shortstack[r]{Leaking\\replies}
      & \scriptsize Blocked
      & \scriptsize\shortstack[r]{Blocked\\after\\delivery}
      & \scriptsize\shortstack[r]{Password\\delivered} \\
    \midrule
    Release each 200-token chunk, then check (default) & 119 & 112 (94\%) & 112 & 119 (100\%) \\
    Release each 20-token chunk, then check & 119 & 117 (98\%) & 87 & 89 (75\%) \\
    Check each 200-token chunk, then release & 120 & 112 (93\%) & 0 & 8 (6.7\%) \\
    \bottomrule
  \end{tabular}
  \par\vspace{3pt}
  \parbox{\columnwidth}{\raggedright Percentages are of leaking replies.  \emph{Blocked after delivery}: blocked leaks whose
  password had already reached the client.  \emph{Password delivered}: leaks
  whose password reached the client, blocked or not.}
\end{table}

\begin{table*}[t]
  \centering
  \footnotesize
  \caption{Four bypasses of a check, the key measurement for each, and
  mitigations.}
  \label{tab:design-changes}
  \renewcommand{\arraystretch}{1.2}
  \setlength{\tabcolsep}{3pt}
  \setlength{\aboverulesep}{0pt}
  \setlength{\belowrulesep}{0pt}
  \setlength{\extrarowheight}{1pt}
  \rowcolors{2}{white}{rowgray}
  \begin{tabular}{>{\raggedright\arraybackslash}p{48pt}
    >{\raggedright\arraybackslash}p{82pt}
    >{\raggedright\arraybackslash}p{88pt}
    >{\raggedright\arraybackslash}p{99pt}
    >{\raggedright\arraybackslash}p{122pt}}
    \toprule\addlinespace[2pt]
    Bypass & Cases & What it leaves open & Key measurement & Mitigation \\
    \midrule
    (a) Unchecked path
      & OpenAI GPT-3.5/GPT-4 fine-tuning; per-account detectors
        (Sec.~\ref{subsec:coverage-results})
      & Attacks on this path avoid the check ($\alpha<1$)
      & Enumeration of paths to the outcome, with the check on each
      & Route every path to the outcome through the check
        (Blacklight~\cite{li2022blacklight}; architectural
        routing~\cite{Debenedetti2025CaMeL,costa2025fides,shi2025progent}) \\
    (b) Premature release
      & Baidu-Chat, DeepSeek, Qwen (Sec.~\ref{subsec:continuation-results});
        NeMo default (Sec.~\ref{subsec:release-measurement})
      & Nearly everything remains after a block \mbox{($r$ near one)}
      & What the client has when the check acts, scored on the
        evaluated outcome (Sec.~\ref{subsec:release-measurement})
      & Limit what the client has received when the check acts \\
    (c) Observable actions after a block
      & CaMeL's NORMAL mode (Sec.~\ref{subsec:continuation-results})
      & Something remains after a block ($r>0$)
      & Noninterference guarantee for the protected data (Fides gives one
        for integrity)
      & Make attacker-observable actions after the check, including control
        flow, independent of protected data \\
    (d) Attempts across invocations
      & SmoothLLM, sequential monitor: one attempt (Secs.~\ref{subsec:continuation-results},~\ref{subsec:failure-results})
      & The most an attacker can obtain grows with allowed attempts
        (Sec.~\ref{subsec:deployment-composition})
      & Coverage and miss rate after every history; what remains
        when all attempts are blocked
      & Set stopping rules for retries and evaluate over the attempts
        they allow (Eq.~\ref{eq:retry-mediation-bound}) \\
    \bottomrule
  \end{tabular}
\end{table*}

\subsection{Certificates from All Three Properties}
\label{subsec:upper-endpoint-results}

Fides's formal configuration certifies that no harmful assistance
remains for an \emph{integrity outcome}, an action the service must not
take~\cite{costa2025fides}.  We adopt the source's enforcing agent loop,
its taint-tracking planner, and its assumptions: deterministic model
semantics and conditions on trusted configuration, labels, and tools.
The outcome is whether changing only the untrusted initial data alters
an observation of trusted state in the source's noninterference comparison.

The integrity policy checks every tool that writes trusted state and
every argument that affects it, giving full coverage.  Taint tracking
gives a zero miss rate: a required check allows a write only when taint
tracking labels every input as trusted.  After a block, the noninterference result leaves this outcome unreachable
on any fixed finite prefix, so no later step reaches it and nothing remains after a block.  With all three
properties established, no harmful assistance remains for this outcome.

The source's confidentiality result, \emph{explicit secrecy}, guarantees
that confidential data never reaches an unauthorized reader through an
explicit flow; the guarantee does not cover implicit flows caused by control flow.  A
check on every tool call that sends data gives full coverage, taint
tracking again gives a zero miss rate, and the guarantee holds over the
whole run, so nothing remains after a block.  This evaluation also
establishes all three properties, so no harmful assistance remains for this
outcome either.  Either conclusion also holds for another deployment that
keeps the same evaluation setting and these three properties, or whose scope
lies within the scope of these results.

\section{Experiment on an Open-Source Guardrail}
\label{sec:implications}
\label{subsec:release-measurement}

This section asks what the block rate, the metric that evaluations usually
report, misses on a guardrail as deployed.  None of the 16 input and output
classifier evaluations reports what remains after a block
(Table~\ref{tab:evidence-profile}).  We measure the block rate and what
remains after a block on NeMo Guardrails 0.24.1~\cite{rebedea2023nemo}, an open-source guardrail, using
Qwen-Plus as the model.  For streamed output, NeMo Guardrails by default
releases each 200-token chunk to the client before checking it, so the
check can block a reply after the password has already reached the client.
The output check asks the same model whether a reply violates a policy: NeMo
Guardrails' example policy plus one rule against revealing the password.
All replies on the evaluated output path pass through this check, so
coverage on that path is full.  We drew 200 attacks at random from the
Tensor Trust extraction benchmark~\cite{toyer2024tensortrust}, each with its
own defense prompt and password, and ran each release order once on all 200
attacks, generating replies at temperature~0.  Following Tensor Trust, a model
reply is scored as leaking when it contains the password.
For comparison, Table~\ref{tab:nemo-release} also reports 20-token chunks
and checking each chunk before release; the check is the same in all three.

In the default configuration, the check blocked 94\% of leaking replies,
the figure that an evaluation of the check would report.  Yet the client
had already received the password in every blocked reply, so the password
reached the client in all 119 leaking replies and everything remained after
a block ($r=1$).

Across the three release orders, the same check blocked between 93\% and
98\% of leaking replies, while the share of leaking replies that delivered
the password to the client was 100\%, 75\%, and 6.7\%.  What changed is
what remained after a block, and the block rate does not reflect it.

\section{Implications for Deployment}
\label{subsec:governance}

This section asks what evaluators should measure, and what deployments can
do, for four ways that published cases show a check can be bypassed.
Table~\ref{tab:design-changes} maps the cases of
Sections~\ref{sec:deployment} and~\ref{sec:implications} to four bypasses of
a check, each with the key measurement that assesses it and a mitigation.
An enumeration of the paths to the outcome, with the check on each, shows
whether unchecked paths (a) exist; routing every path to the outcome through
the check prevents them.  Premature release (b) and observable
actions after a block (c) leave assistance obtainable despite detection.
Bypass (b) is assessed by what the client has received when the check
acts, scored on the evaluated outcome.  Bypass (c) is assessed by a
noninterference guarantee for the protected data; information-flow controls
can prevent bypass (c) on the paths and observations they cover.  After a
block, an attacker may submit another request or start a new conversation
(d) if the deployment allows it.  To show the most these attempts can
obtain, the deployment needs evidence that the check keeps its coverage and
miss rate on every attempt after any allowed earlier attempts, and evidence
on what remains when the check blocks every allowed attempt.

Once every path passes through the check and nothing remains after a block,
the most an attacker can obtain in one attempt equals the miss rate that
holds against every attacker in the threat model, so each reduction in this miss rate lowers the most an attacker can obtain by the same amount.
Limiting the allowed attempts also reduces the most an attacker can obtain
with the check unchanged (Section~\ref{subsec:deployment-composition}).

\section{Conclusion}
\label{sec:conclusion}

For a safeguard that blocks attacks with a check, the key measurements of
its effect in deployment are coverage, the miss rate against every attacker
in the threat model, and what remains after a block.  A test-suite block
rate measures the miss rate only on the tested attacks; without the other
two, even a low miss rate cannot rule out that an attacker obtains all the
harmful assistance.  Stacked checks and repeated attempts need evidence
on how they fail together: layers measured separately guarantee only the
strongest layer's miss rate, and the most an attacker can obtain grows with
the allowed attempts.

Coverage and what remains after a block depend on the deployment's paths
and release of output, so they must be established on the system as deployed.
In NeMo Guardrails' default streaming configuration, the check blocked
94\% of password leaks, yet every blocked leak had already reached the
client, and the block rate did not reflect this.

Only 2 of 85 published safeguard evaluations provide evidence on all
three; the others can show improvements under the conditions they
test, but not a safer deployed system.  Deployers who require these three
measurements on the deployed system before launch can find where attackers
still obtain harmful assistance and change the deployment there, so a
safeguard that works in evaluation also makes the deployed system safer.

\section*{Acknowledgment}
This work was supported by the New Generation Artificial Intelligence National Science and Technology Major Project under Grant 2025ZD0123202, the National Natural Science Foundation of China under Grant U2336206, and the Anhui Province Key Laboratory of Digital Security.
The authors used Claude Opus~5.5 for translation and language polishing of
this manuscript.  Claude Opus~5.5 and GPT-5.6 Sol were the screening and coding
models in the literature measurement of Section~\ref{sec:theory-guided-systematization}.

\bibliographystyle{IEEEtran}
\bibliography{reference_controls,references_tifs}

\begin{thebibliography}{10}
\providecommand{\url}[1]{#1}
\csname url@samestyle\endcsname
\providecommand{\newblock}{\relax}
\providecommand{\bibinfo}[2]{#2}
\providecommand{\BIBentrySTDinterwordspacing}{\spaceskip=0pt\relax}
\providecommand{\BIBentryALTinterwordstretchfactor}{4}
\providecommand{\BIBentryALTinterwordspacing}{\spaceskip=\fontdimen2\font plus
\BIBentryALTinterwordstretchfactor\fontdimen3\font minus
  \fontdimen4\font\relax}
\providecommand{\BIBforeignlanguage}[2]{{%
\expandafter\ifx\csname l@#1\endcsname\relax
\typeout{** WARNING: IEEEtran.bst: No hyphenation pattern has been}%
\typeout{** loaded for the language `#1'. Using the pattern for}%
\typeout{** the default language instead.}%
\else
\language=\csname l@#1\endcsname
\fi
#2}}
\providecommand{\BIBdecl}{\relax}
\BIBdecl

\bibitem{solaiman2019release}
I.~Solaiman \emph{et~al.}, ``Release strategies and the social impacts of
  language models,'' 2019, arXiv:1908.09203.

\bibitem{li2024wmdp}
N.~Li \emph{et~al.}, ``The {WMDP} benchmark: Measuring and reducing malicious
  use with unlearning,'' in \emph{Proc. ICML}, vol. 235, 2024, pp.
  28\,525--28\,550.

\bibitem{DBLP:conf/uss/ZhanCSS25}
X.~Zhan, J.~C. Carrillo, W.~Seymour, and J.~Such, ``Malicious {LLM}-based
  conversational {AI} makes users reveal personal information,'' in \emph{Proc.
  USENIX Security}, 2025, pp. 61--80.

\bibitem{OpenAI2026GPT56SystemCard}
\BIBentryALTinterwordspacing
{OpenAI}, ``{GPT-5.6} system card,'' OpenAI Deployment Safety Hub, Jul. 9,
  2026, {Accessed}: Aug. 14, 2026. [Online]. Available:
  \url{https://deploymentsafety.openai.com/gpt-5-6}
\BIBentrySTDinterwordspacing

\bibitem{bai2022constitutional}
Y.~Bai \emph{et~al.}, ``Constitutional {AI}: Harmlessness from {AI} feedback,''
  2022, arXiv:2212.08073.

\bibitem{dai2024saferlhf}
\BIBentryALTinterwordspacing
J.~Dai \emph{et~al.}, ``Safe {RLHF}: Safe reinforcement learning from human
  feedback,'' in \emph{Proc. ICLR}, 2024. [Online]. Available:
  \url{https://proceedings.iclr.cc/paper_files/paper/2024/hash/dd1577afd396928ed64216f3f1fd5556-Abstract-Conference.html}
\BIBentrySTDinterwordspacing

\bibitem{arditi2024refusal}
A.~Arditi \emph{et~al.}, ``Refusal in language models is mediated by a single
  direction,'' in \emph{Proc. NeurIPS}, vol.~37, 2024.

\bibitem{zou2024circuitbreakers}
A.~Zou \emph{et~al.}, ``Improving alignment and robustness with circuit
  breakers,'' in \emph{Proc. NeurIPS}, vol.~37, 2024.

\bibitem{inan2023llamaguard}
H.~Inan \emph{et~al.}, ``{Llama Guard}: {LLM}-based input-output safeguard for
  human-{AI} conversations,'' 2023, arXiv:2312.06674.

\bibitem{han2024wildguard}
S.~Han \emph{et~al.}, ``{WildGuard}: Open one-stop moderation tools for safety
  risks, jailbreaks, and refusals of {LLMs},'' in \emph{Proc. NeurIPS},
  vol.~37, 2024.

\bibitem{Cunningham2026ConstitutionalClassifiersPP}
\BIBentryALTinterwordspacing
H.~Cunningham \emph{et~al.}, ``Constitutional classifiers++: Efficient
  production-grade defenses against universal jailbreaks,'' in \emph{Proc.
  ICLR}, 2026. [Online]. Available:
  \url{https://proceedings.iclr.cc/paper_files/paper/2026/hash/4f697305ef1f868ad77c3c0027989a6f-Abstract-Conference.html}
\BIBentrySTDinterwordspacing

\bibitem{Anthropic2025MisuseReport}
\BIBentryALTinterwordspacing
{Anthropic}, ``Detecting and countering misuse of {AI}: August 2025,''
  Anthropic Threat Intelligence Report, Aug. 27, 2025, {Accessed}: Aug. 20,
  2026. [Online]. Available:
  \url{https://www.anthropic.com/news/detecting-countering-misuse-aug-2025}
\BIBentrySTDinterwordspacing

\bibitem{OpenAI2025DisruptingMaliciousUses}
\BIBentryALTinterwordspacing
{OpenAI}, ``Disrupting malicious uses of {AI}: An update ({October} 2025),''
  OpenAI Threat Intelligence Report, Oct. 7, 2025, {Accessed}: Aug. 20, 2026.
  [Online]. Available:
  \url{https://openai.com/global-affairs/disrupting-malicious-uses-of-ai-october-2025/}
\BIBentrySTDinterwordspacing

\bibitem{Debenedetti2025CaMeL}
E.~Debenedetti \emph{et~al.}, ``Defeating prompt injections by design,'' in
  \emph{Proc. IEEE SaTML}, 2026, pp. 587--618.

\bibitem{costa2025fides}
M.~Costa \emph{et~al.}, ``Securing {AI} agents with information-flow control,''
  2025, arXiv:2505.23643.

\bibitem{shi2025progent}
T.~Shi \emph{et~al.}, ``{Progent}: Securing {AI} agents with privilege
  control,'' 2025, arXiv:2504.11703.

\bibitem{wu2025isolategpt}
Y.~Wu, F.~Roesner, T.~Kohno, N.~Zhang, and U.~Iqbal, ``{IsolateGPT}: An
  execution isolation architecture for {LLM}-based agentic systems,'' in
  \emph{Proc. NDSS}, 2025.

\bibitem{Anthropic2026Containment}
\BIBentryALTinterwordspacing
{Anthropic}, ``How we contain {Claude} across products,'' Anthropic
  Engineering, May 25, 2026, {Accessed}: Aug. 14, 2026. [Online]. Available:
  \url{https://www.anthropic.com/engineering/how-we-contain-claude}
\BIBentrySTDinterwordspacing

\bibitem{mazeika2024harmbench}
M.~Mazeika \emph{et~al.}, ``{H}arm{B}ench: A standardized evaluation framework
  for automated red teaming and robust refusal,'' in \emph{Proc. ICML}, vol.
  235, 2024, pp. 35\,181--35\,224.

\bibitem{chao2024jailbreakbench}
P.~Chao \emph{et~al.}, ``{JailbreakBench}: An open robustness benchmark for
  jailbreaking large language models,'' in \emph{Proc. NeurIPS}, vol.~37, 2024.

\bibitem{rottger2024xstest}
P.~R{\"o}ttger, H.~Kirk, B.~Vidgen, G.~Attanasio, F.~Bianchi, and D.~Hovy,
  ``{XSTest}: A test suite for identifying exaggerated safety behaviours in
  large language models,'' in \emph{Proc. NAACL}, 2024, pp. 5377--5400.

\bibitem{debenedetti2024agentdojo}
E.~Debenedetti, J.~Zhang, M.~Balunovic, L.~Beurer-Kellner, M.~Fischer, and
  F.~Tram{\`e}r, ``{AgentDojo}: A dynamic environment to evaluate prompt
  injection attacks and defenses for {LLM} agents,'' in \emph{Proc. NeurIPS
  Datasets and Benchmarks Track}, vol.~37, 2024, pp. 82\,895--82\,920.

\bibitem{peppin2025reality}
A.~Peppin \emph{et~al.}, ``The reality of {AI} and biorisk,'' in \emph{Proc.
  ACM FAccT}, 2025, pp. 763--771.

\bibitem{lukosiute2025cyberrisk}
K.~Lukošiūtė and A.~Swanda, ``{LLM} cyber evaluations don't capture
  real-world risk,'' 2025, arXiv:2502.00072.

\bibitem{CarliniWagner2017NotEasilyDetected}
N.~Carlini and D.~Wagner, ``Adversarial examples are not easily detected:
  Bypassing ten detection methods,'' in \emph{Proc. ACM AISec}, 2017, pp.
  3--14.

\bibitem{Tramer2020AdaptiveAttacks}
F.~Tram{\`e}r, N.~Carlini, W.~Brendel, and A.~Madry, ``On adaptive attacks to
  adversarial example defenses,'' in \emph{Proc. NeurIPS}, 2020.

\bibitem{Nasr2026AttackerMovesSecond}
M.~Nasr \emph{et~al.}, ``The attacker moves second: Stronger adaptive attacks
  bypass defenses against {LLM} jailbreaks and prompt injections,'' in
  \emph{Proc. USENIX Security}, 2026, pp. 1467--1486.

\bibitem{qi2025durability}
X.~Qi \emph{et~al.}, ``On evaluating the durability of safeguards for
  open-weight {LLM}s,'' in \emph{Proc. ICLR}, 2025.

\bibitem{Glukhov2025BreachLeaks}
D.~Glukhov, Z.~Han, I.~Shumailov, V.~Papyan, and N.~Papernot, ``Breach by a
  thousand leaks: Unsafe information leakage in `safe' {AI} responses,'' in
  \emph{Proc. ICLR}, 2025.

\bibitem{SaltzerSchroeder1975}
\BIBentryALTinterwordspacing
J.~H. Saltzer and M.~D. Schroeder, ``The protection of information in computer
  systems,'' \emph{Proc. IEEE}, vol.~63, no.~9, pp. 1278--1308, 1975. [Online].
  Available: \url{https://doi.org/10.1109/PROC.1975.9939}
\BIBentrySTDinterwordspacing

\bibitem{Schneider2000Enforceable}
\BIBentryALTinterwordspacing
F.~B. Schneider, ``Enforceable security policies,'' \emph{ACM Trans. Inf. Syst.
  Secur.}, vol.~3, no.~1, pp. 30--50, 2000. [Online]. Available:
  \url{https://doi.org/10.1145/353323.353382}
\BIBentrySTDinterwordspacing

\bibitem{Clymer2025SafetyCase}
J.~Clymer, J.~Weinbaum, R.~Kirk, K.~Mai, S.~Zhang, and X.~Davies, ``An example
  safety case for safeguards against misuse,'' 2025, arXiv:2505.18003.

\bibitem{AISISafeguardsTeam2025Principles}
\BIBentryALTinterwordspacing
{UK AI Safety Institute Safeguards Analysis Team}, ``Principles for evaluating
  misuse safeguards of frontier {AI} systems,'' UK AI Safety Institute, Tech.
  Rep., Feb. 4, 2025, organization renamed the {AI Security Institute} on Feb.
  14, 2025. {Accessed}: Aug. 14, 2026. [Online]. Available:
  \url{https://www.aisi.gov.uk/blog/principles-for-safeguard-evaluation}
\BIBentrySTDinterwordspacing

\bibitem{Jain2026AdaptiveAdversaries}
D.~Jain, D.~Hartmann, and C.~Li, ``Adaptive adversaries: A multi-turn,
  multi-{LLM} benchmark for {LLM} agent security,'' 2026, arXiv:2607.18063.

\bibitem{Rao2026FragFuse}
Z.~Rao \emph{et~al.}, ``{FragFuse}: Bypassing access control of large language
  model agents via memory-based query fragmentation and fusion,'' in
  \emph{Proc. USENIX Security}, 2026, pp. 4761--4780.

\bibitem{Mahato2026SafeguardConditionedUplift}
D.~T. Mahato, ``Safeguard-conditioned uplift: Measuring utility--risk frontiers
  for dual-use biology assistants,'' 2026, arXiv:2607.13039.

\bibitem{Wu2026CopyableContext}
P.~Wu, W.~Zhang, and N.~Yu, ``{LLM} misuse safeguards rely on untrusted input:
  {A} widespread, high-risk vulnerability,'' 2026, arXiv:2607.27951.

\bibitem{Athalye2018ObfuscatedGradients}
A.~Athalye, N.~Carlini, and D.~Wagner, ``Obfuscated gradients give a false
  sense of security: Circumventing defenses to adversarial examples,'' in
  \emph{Proc. ICML}, vol.~80, 2018, pp. 274--283.

\bibitem{Anthropic2026Fable}
\BIBentryALTinterwordspacing
{Anthropic}, ``Introducing {Claude} {Opus} 5.5,'' Model announcement, Sep. 22,
  2026, {Accessed}: Nov. 17, 2026. [Online]. Available:
  \url{https://www.anthropic.com/claude-opus-5-5}
\BIBentrySTDinterwordspacing

\bibitem{OpenAI2026GPT56}
\BIBentryALTinterwordspacing
{OpenAI}, ``{GPT-5.6},'' Model suite announcement ({Sol}, {Terra}, {Luna}),
  Jul. 9, 2026, {Accessed}: Aug. 20, 2026. [Online]. Available:
  \url{https://openai.com/index/gpt-5-6/}
\BIBentrySTDinterwordspacing

\bibitem{DBLP:conf/aaai/ZhangZLZDTZ26}
Z.~Zhang \emph{et~al.}, ``{FedSEA-LLaMA}: A secure, efficient and adaptive
  federated splitting framework for large language models,'' in \emph{Proc.
  AAAI}, 2026, pp. 28\,680--28\,688.

\bibitem{bhardwaj2024resta}
R.~Bhardwaj, D.~A. Do, and S.~Poria, ``Language models are {Homer Simpson}!
  {Safety} re-alignment of fine-tuned language models through task
  arithmetic,'' in \emph{Proc. ACL}, 2024, pp. 14\,138--14\,149.

\bibitem{DBLP:conf/aaai/BowenMCKGP25}
D.~Bowen, B.~Murphy, W.~Cai, D.~Khachaturov, A.~Gleave, and K.~Pelrine,
  ``Scaling trends for data poisoning in {LLMs},'' in \emph{Proc. AAAI}, 2025,
  pp. 27\,206--27\,214.

\bibitem{li2022blacklight}
H.~Li, S.~Shan, E.~Wenger, J.~Zhang, H.~Zheng, and B.~Y. Zhao, ``{Blacklight}:
  Scalable defense for neural networks against query-based black-box attacks,''
  in \emph{Proc. USENIX Security}, 2022, pp. 2117--2134.

\bibitem{feng2023stateful}
R.~Feng, A.~Hooda, N.~Mangaokar, K.~Fawaz, S.~Jha, and A.~Prakash, ``Stateful
  defenses for machine learning models are not yet secure against black-box
  attacks,'' in \emph{Proc. ACM CCS}, 2023, pp. 786--800.

\bibitem{DBLP:conf/ndss/ZhangLZMC26}
Y.~Zhang, T.~Liu, Z.~Zhao, G.~Meng, and K.~Chen, ``Bleeding pathways: Vanishing
  discriminability in {LLM} hidden states fuels jailbreak attacks,'' in
  \emph{Proc. NDSS}, 2026.

\bibitem{DBLP:conf/acl/SiLBZ26}
W.~M. Si, M.~Li, M.~Backes, and Y.~Zhang, ``Pruning unsafe tickets: A
  resource-efficient framework for safer and more robust {LLMs},'' in
  \emph{Proc. ACL}, 2026, pp. 26\,285--26\,302.

\bibitem{pang2026paladin}
Y.~Pang, W.~Meng, X.~Liao, and T.~Wang, ``{Paladin}: Defending {LLM}-enabled
  phishing emails with a new trigger-tag paradigm,'' in \emph{Proc. NDSS},
  2026.

\bibitem{DBLP:conf/nips/WangHSYLLTHT25}
Y.~Wang \emph{et~al.}, ``Panacea: Mitigating harmful fine-tuning for large
  language models via post-fine-tuning perturbation,'' in \emph{Proc. NeurIPS},
  2025.

\bibitem{DBLP:conf/uss/WangWJL00L0LR25}
X.~Wang \emph{et~al.}, ``{SelfDefend}: {LLMs} can defend themselves against
  jailbreaking in a practical manner,'' in \emph{Proc. USENIX Security}, 2025,
  pp. 2441--2460.

\bibitem{DBLP:journals/corr/abs-2402-07510}
\BIBentryALTinterwordspacing
S.~R. Motwani \emph{et~al.}, ``Secret collusion among {AI} agents: Multi-agent
  deception via steganography,'' in \emph{Proc. NeurIPS}, vol.~37, 2024.
  [Online]. Available:
  \url{https://proceedings.neurips.cc/paper_files/paper/2024/hash/861f7dad098aec1c3560fb7add468d41-Abstract-Conference.html}
\BIBentrySTDinterwordspacing

\bibitem{DBLP:conf/nips/LiLZXSZJPZZY25}
Q.~Li \emph{et~al.}, ``{CoreGuard}: Safeguarding foundational capabilities of
  {LLMs} against model stealing in edge deployment,'' in \emph{Proc. NeurIPS},
  2025.

\bibitem{DBLP:conf/ndss/AbloveCQE26}
A.~Ablove, S.~Chandrashekaran, X.~Qiang, and R.~Ensafi, ``Characterizing the
  implementation of censorship policies in {Chinese} {LLM} services,'' in
  \emph{Proc. NDSS}, 2026.

\bibitem{chen2026monitoring}
\BIBentryALTinterwordspacing
Y.-H. Chen, N.~Joshi, Y.~Chen, M.~Andriushchenko, R.~Angell, and H.~He,
  ``Monitoring decomposition attacks with lightweight sequential monitors,'' in
  \emph{Proc. ICLR}, 2026. [Online]. Available:
  \url{https://proceedings.iclr.cc/paper_files/paper/2026/hash/2985777c8f2152664946ed30deb77291-Abstract-Conference.html}
\BIBentrySTDinterwordspacing

\bibitem{robey2024smoothllmv4}
\BIBentryALTinterwordspacing
A.~Robey, E.~Wong, H.~Hassani, and G.~J. Pappas, ``{SmoothLLM}: Defending large
  language models against jailbreaking attacks,'' \emph{Trans. Mach. Learn.
  Res.}, 2025. [Online]. Available:
  \url{https://openreview.net/forum?id=laPAh2hRFC}
\BIBentrySTDinterwordspacing

\bibitem{rebedea2023nemo}
T.~Rebedea, R.~Dinu, M.~N. Sreedhar, C.~Parisien, and J.~Cohen, ``{NeMo
  Guardrails}: A toolkit for controllable and safe {LLM} applications with
  programmable rails,'' in \emph{Proc. EMNLP: System Demonstrations}, 2023, pp.
  431--445.

\bibitem{toyer2024tensortrust}
\BIBentryALTinterwordspacing
S.~Toyer \emph{et~al.}, ``{Tensor Trust}: Interpretable prompt injection
  attacks from an online game,'' in \emph{Proc. ICLR}, 2024. [Online].
  Available: \url{https://openreview.net/forum?id=fsW7wJGLBd}
\BIBentrySTDinterwordspacing

\end{thebibliography}

\appendices
\section{Proofs}
\label{app:deployment-theory}

This appendix states the formal setting, proves the results of
Section~\ref{sec:deployment-theory},
and, for each row of
Table~\ref{tab:exchange-rate}, gives a deployment that reaches the stated
amount exactly.  Probability distributions are defined on the
declared trajectory space; all spaces are standard Borel and all declared
events are measurable.  Finite or countable structure is assumed only where a
result states it.

\subsection{Formal Setting}
\label{app:formal-setting}

We write the six conditions of the evaluation setting in
Section~\ref{sec:definitions-scope} as $S$ (the service without the
intervention), $D$ (the intervention), $\Sigma$ (the threat model),
$\mathcal L$ (the basic-rights and legitimate-use constraints), $Z$ (the
outcome scale), and $\mathcal E$ (the task, environment, and horizon).
Let $\Omega_D$ be the set of complete security-relevant trajectories under
$D$; each $\omega\in\Omega_D$ records the evidence available before each
decision.  Declaring $Z$ also fixes a projection
$\phi_Z:\Omega_D\rightarrow\mathcal T_Z$, written $T=\phi_Z(\omega)$ and
called the trace.
The projection retains the information needed to compute $v_Z(T)\in[0,1]$,
the most harmful assistance obtainable from $T$.

A causal policy $g\in\mathcal G_D$ implements $D$ and satisfies every
condition in $\mathcal L$.  The deployment commits to this policy before
the attacker chooses a strategy $\sigma\in\Sigma$.
Let $P_{\sigma}^{g}$ be the distribution of trajectories on
$\Omega_D$ that attack $\sigma$ induces; Section~\ref{sec:definitions-scope}
writes it $P_\sigma$, with expectation $\mathbb E_\sigma$, and writes $v$
for $v_Z(T)$.  Randomized attacks and their
total-variation limits form the set of distributions available to the attacker,
\begin{equation}
  \mathcal C_g=
  \overline{\operatorname{co}}
  \{P_{\sigma}^{g}:\sigma\in\Sigma\}.
  \label{eq:attacker-trace-class}
\end{equation}
The harmful assistance that remains is the largest expected
harmful assistance over this set,
\begin{equation}
  W_Z(g)=
  \sup_{Q\in\mathcal C_g}
  \mathbb E_Q[v_Z(T)],
  \label{eq:closed-hull-value}
\end{equation}
and for a committed deployment it equals $Z(S[D])$ from
Section~\ref{sec:definitions-scope}.

We use one elementary fact.  For any measurable $f$ with range
in $[0,1]$,
\begin{equation}
  |\mathbb E_Pf-\mathbb E_Qf|
  \leq\operatorname{TV}(P,Q).
  \label{eq:tv-bounded-expectation}
\end{equation}

Section~\ref{sec:deployment-theory} quantifies over strategies
$\sigma\in\Sigma$; the proofs below quantify over distributions
$Q\in\mathcal C_g$.  The two forms agree.  For measurable $f$ with range in
$[0,1]$, the map $Q\mapsto\mathbb E_Qf$ is affine and, by
Equation~\ref{eq:tv-bounded-expectation}, continuous in total variation.
Convexification and total-variation closure therefore change neither its
supremum nor its infimum, so
\begin{equation}
  W_Z(g)=\sup_{\sigma\in\Sigma}
  \mathbb E_{P_{\sigma}^{g}}[v_Z(T)],
  \label{eq:strategy-form}
\end{equation}
the largest expected score over the strategies in the threat model, and
$\inf_{Q\in\mathcal C_g}Q(A)
=\inf_{\sigma\in\Sigma}P_{\sigma}^{g}(A)$
for every event $A$.  For events $A,F\subseteq\Omega_D$ and
$\epsilon\in[0,1]$, the maps $Q\mapsto Q(F\cap A)-\epsilon Q(A)$ and
$Q\mapsto Q(\mathcal R)$, for a set $\mathcal R\subseteq\Omega_D$ of
trajectories, are also affine and continuous, so a constraint
$Q(F\cap A)\leq\epsilon Q(A)$ or $Q(\mathcal R)=1$ holds for every
$Q\in\mathcal C_g$ if and only if it holds for every
$P_{\sigma}^{g}$ with $\sigma\in\Sigma$.

The distributions the attacker can reach
are the randomized attacks, the convex hull of
$\{P_{\sigma}^{g}:\sigma\in\Sigma\}$.  The set $\mathcal C_g$ adds
their total-variation limits, which by the continuity above changes none of
the quantities in Section~\ref{sec:deployment-theory}.

\subsection{Executed Attacks, Reachability, and Enumerated Attacks}
\label{app:direct-bounds}

By Equation~\ref{eq:strategy-form}, every $\sigma\in\Sigma$ gives the lower
bound
\begin{equation}
  W_Z(g)\geq\mathbb E_{P_{\sigma}^{g}}[v_Z(T)].
  \label{eq:attack-witness-bound}
\end{equation}
If a set $\mathcal M$ of distributions contains every distribution the
attacker can reach, then the supremum of
$Q\mapsto\mathbb E_Q[v_Z(T)]$ over $\mathcal M$ is at least
its supremum over these distributions, which equals $W_Z(g)$ by the
continuity above, so
\begin{equation}
  W_Z(g)\leq
  \sup_{Q\in\mathcal M}
    \mathbb E_Q[v_Z(T)],
  \label{eq:direct-reachable-upper}
\end{equation}
with equality when $\mathcal M\subseteq\mathcal C_g$.

For the enumerated-attack entry of Table~\ref{tab:exchange-rate}, a finite set
of tested modifications
$\{\sigma_1,\ldots,\sigma_n\}\subseteq\Sigma$ generates distributions inside
$\mathcal C_g$.  For any candidate upper bound $c<1$, consider a $\Sigma$ that
also contains an untested strategy whose distribution places all mass on a trace with $v_Z=1$.  The enumeration is unchanged, yet $W_Z(g)=1>c$.  At any sample size,
the results on a sample of $\Sigma$ therefore cannot exclude that the
attacker obtains all the harmful assistance.  If the
sample is shown to exhaust $\Sigma$, it is no longer a sample, and the
equality case of Equation~\ref{eq:direct-reachable-upper} applies.

\subsection{Blocking Events}
\label{app:blocking-proof}

Let $G\subseteq\Omega_D$ be the trajectories on which the safeguard applies
its blocking response, and $\mathcal R_g\subseteq\Omega_D$ a set with
$Q(\mathcal R_g)=1$ for every $Q\in\mathcal C_g$; we call the trajectories in
$\mathcal R_g$ attacker-reachable.  Define
\begin{equation}
  \lambda_G=\inf_{Q\in\mathcal C_g}Q(G),
  \qquad
  r_G=\sup_{\omega\in G\cap\mathcal R_g}
       v_Z(T(\omega)),
  \label{eq:success-region-parameters}
\end{equation}
using $r_G=0$ when $G\cap\mathcal R_g$ is empty.
Every committed policy satisfies the upper bound
\begin{equation}
  W_Z(g)\leq 1-\lambda_G(1-r_G),
  \label{eq:success-region-sharp}
\end{equation}
because every attacker distribution reaches $G$ with probability at least
$\lambda_G$, scores at most $r_G$ there, and at most one elsewhere.  The bound is tight given $\lambda_G$ and $r_G$: a distribution that places
mass $\lambda_G$ on value $r_G$ inside $G$ and the remaining mass on value one outside
$G$ attains it.

The right side of Equation~\ref{eq:success-region-sharp} is decreasing in
$\lambda_G$ and increasing in $r_G$.  Any $\lambda\leq\lambda_G$ and any
$r\geq r_G$ therefore yield the bound of
Section~\ref{subsec:success-region},
\begin{equation}
  W_Z(g)\leq 1-\lambda(1-r),
  \label{eq:success-region-bound}
\end{equation}
whose premises are exactly these two inequalities.  The same construction with
mass $\lambda$ and value $r$ has $\lambda_G=\lambda$ and $r_G=r$, so no
tighter bound follows from $\lambda$ and $r$ alone.

\subsection{Proof of Theorem~\ref{thm:local-contract-lift}}
\label{app:closed-mediation-proof}

Let $A\subseteq\Omega_D$ be the event that the attack path enters the part
of the system the check inspects, and $F$ the event that the check misses
the attack there; $F^c$ is the event that it does not miss, so $A\cap F^c$
is the event that the check inspects and stops the attack.  The miss rate
and the coverage of Theorem~\ref{thm:local-contract-lift} are
\begin{equation}
  \begin{gathered}
  P_{\sigma}^{g}(F\cap A)\leq\epsilon P_{\sigma}^{g}(A)
  \ \ \text{for every }\sigma\in\Sigma,\\
  \alpha=\inf_{\sigma\in\Sigma}P_{\sigma}^{g}(A),
  \end{gathered}
  \label{eq:robust-mediation-contract}
\end{equation}
and its premise on what remains after a block is $v_Z(T(\omega))\leq r$ for
every $\omega\in A\cap F^c\cap\mathcal R_g$.  The theorem states
$W_Z(g)\leq1-\alpha(1-\epsilon)(1-r)$.

\begin{proof}[Proof of Theorem~\ref{thm:local-contract-lift}]
For any $Q\in\mathcal C_g$, boundedness of $v_Z$ and the premise on $A\cap F^c$ give
\begin{align*}
  \mathbb E_Qv_Z
  &\leq rQ(A\cap F^c)+1-Q(A\cap F^c)\\
  &=1-(1-r)Q(A\cap F^c).
\end{align*}
Both parts of Equation~\ref{eq:robust-mediation-contract}, which hold for
$Q$ by Appendix~\ref{app:formal-setting}, imply
\[
  Q(A\cap F^c)
  =Q(A)-Q(A\cap F)
  \geq(1-\epsilon)Q(A)
  \geq\alpha(1-\epsilon).
\]
Substitution and maximization over $Q$ prove the bound.

The bound is tight from exactly these three quantities.  On a three-point
space, assign mass $\alpha(1-\epsilon)$ to $A\cap F^c$ with value $r$, mass
$\alpha\epsilon$ to $A\cap F$ with value one, and mass $1-\alpha$ to $A^c$
with value one.  All three premises hold
with equality, and the expectation is $1-\alpha(1-\epsilon)(1-r)$.  If
the attacker can reach only this distribution, $W_Z(g)$ equals the bound.

The tight bound is below one exactly when $\alpha(1-\epsilon)(1-r)>0$, which
is equivalent to the three strict conditions $\alpha>0$, $\epsilon<1$, and
$r<1$.  This proves both the if-and-only-if
statement and the impossibility of a stronger uniform bound.
\end{proof}

Lowering the miss rate from $\epsilon$ to $\epsilon'$ therefore reduces the
most an attacker can obtain by exactly $\alpha(\epsilon-\epsilon')(1-r)$, as
stated in Section~\ref{subsec:closed-mediation}.
Even with $\epsilon=0$, the attacker in the construction above obtains all
the harmful assistance when $\alpha=0$ or $r=1$.  In the $r=1$ case, $A$ may be taken
to be the whole space and $F$ empty: the check then inspects and stops every
attack, yet every stopped attack scores one ($W_Z(g)=1$).

The three quantities show that no harmful assistance remains exactly when
$1-\alpha(1-\epsilon)(1-r)=0$, which requires
$\alpha(1-\epsilon)(1-r)=1$.  Since each factor lies in $[0,1]$, this holds
exactly when $\alpha=1$, $\epsilon=0$, and $r=0$.  If any of the three
fails, the attacker in the construction above obtains some harmful
assistance, so these quantities alone cannot show that none remains.

\subsection{Composition Bounds}
\label{app:composition-proofs}

Let $F_j$ be the event that layer~$j$ of an $m$-layer stack misses.  If
$\Pr(F_j\mid F_1,\ldots,F_{j-1},h)\leq\epsilon_j$ after every history $h$
that reaches layer~$j$, then
\begin{equation}
  \Pr\!\left(\bigcap_{j=1}^{m}F_j\right)\leq\prod_j\epsilon_j.
  \label{eq:serial-product-bound}
\end{equation}
With only the marginal bounds $\Pr(F_j)\leq\epsilon_j$,
\begin{equation}
  \Pr\!\left(\bigcap_{j=1}^{m}F_j\right)\leq\min_j\epsilon_j.
  \label{eq:serial-marginal-bound}
\end{equation}

\begin{proof}[Proof of the composition bounds]
For a serial path, the chain rule for conditional probability gives
\[
  \Pr\!\left(\bigcap_{j=1}^{m}F_j\right)
  =\Pr(F_1)
   \prod_{j=2}^{m}
   \Pr\!\left(F_j\mid\bigcap_{i<j}F_i\right).
\]
The premise of Equation~\ref{eq:serial-product-bound} bounds every factor,
which proves the bound without an independence assumption.  Independent layers attain it.

With marginal bounds only, $\Pr(\cap_jF_j)\leq\Pr(F_i)\leq\epsilon_i$ for
each $i$; taking the minimum proves
Equation~\ref{eq:serial-marginal-bound}.  The bound is
attained when all events share a common subevent of probability
$\min_i\epsilon_i$, so a stack of any depth whose layers fail together
misses as often as its strongest layer.

For Equation~\ref{eq:retry-mediation-bound}, the premises of
Equation~\ref{eq:robust-mediation-contract} hold for each of the $H$
attempts after any earlier history, so each attempt is blocked with
conditional probability at least $\alpha(1-\epsilon)$.  The same chain rule
shows that all $H$ attempts are blocked with probability at least
$[\alpha(1-\epsilon)]^H$.  Setting $G$ in
Equation~\ref{eq:success-region-sharp} to the event that all $H$
attempts are blocked gives $\lambda_G\geq[\alpha(1-\epsilon)]^H$ and
$r_G\leq r_0$, which proves Equation~\ref{eq:retry-mediation-bound}.
\end{proof}

\section{Search and Screening}
\label{app:search}

\subsection{Sources and Search}

We search three sources: arXiv, restricted to
\texttt{cs.CR}, \texttt{cs.CL}, \texttt{cs.AI}, \texttt{cs.LG},
\texttt{cs.SE}, \texttt{cs.MA}, and \texttt{stat.ML}; the ACL
Anthology; and dblp, where we enumerate proceedings of the major security and
machine-learning venues and their workshop volumes with a title prefilter.  All three sources cover work through August 2026, from 2019 for arXiv and
the ACL Anthology and from 2023 for dblp.

\paragraph{Query terms}
The query terms follow the structure of a safety claim about
a deployed system (Section~\ref{sec:definitions-scope}):
\emph{[intervention] reduces / prevents / bounds [adverse outcome] for
[deployed system] under [attacker or usage conditions]}.  We extract three facets:
facet~\textbf{A}, the deployed system (\texttt{LLM},
\texttt{foundation model}, \texttt{LLM agent}, and variants);
facet~\textbf{B}, the claim verb or evidential noun (\texttt{defense},
\texttt{mitigate}, \texttt{safeguard}, \texttt{guarantee}, plus adjudication
nouns such as \texttt{benchmark} and \texttt{red teaming}); and
facet~\textbf{C}, the adverse outcome (\texttt{misuse}, \texttt{jailbreak},
\texttt{prompt injection}, \texttt{exfiltration}, and related terms).  A
record is a candidate when it matches \texttt{A AND (B OR C)}.  A record
is also a candidate when it matches facet~\textbf{D}, a lexicon of mechanism
terms that retrieves records matching neither \textbf{B} nor \textbf{C}.

\subsection{Deduplication}

A preprint, its conference paper, and a journal extension form one family.  A workshop
paper and its later full version form one family when the contribution is
the same, and two when the later claim set differs materially.

\begin{table*}[t]
  \centering
  \footnotesize
  \caption{Evaluations counted in the last three columns of
  Table~\ref{tab:evidence-profile}.}
  \label{tab:evidence-list}
  \renewcommand{\arraystretch}{1.2}
  \setlength{\tabcolsep}{4pt}
  \begin{tabular*}{\textwidth}{@{\extracolsep{\fill}}lcccl@{}}
    \toprule
    Evaluation & \shortstack{Guaranteed\\miss rate} & \shortstack{Full\\coverage}
      & \shortstack{Evidence\\after a block} & Section \\
    \midrule
    \multicolumn{5}{@{}l}{\emph{Capability removal}} \\
    WMDP~\cite{li2024wmdp}, RMU unlearning served through an API
      & \xmark & \cmark & \xmark & \ref{subsec:evidence-overview} \\
    \multicolumn{5}{@{}l}{\emph{System and agent architecture}} \\
    FedSEA-LLaMA~\cite{DBLP:conf/aaai/ZhangZLZDTZ26}, input reconstruction scored by ROUGE-2
      & \xmark & \cmark & \xmark & \ref{subsec:evidence-overview} \\
    FedSEA-LLaMA~\cite{DBLP:conf/aaai/ZhangZLZDTZ26}, input reconstruction scored by BLEU-4
      & \xmark & \cmark & \xmark & \ref{subsec:evidence-overview} \\
    CoreGuard~\cite{DBLP:conf/nips/LiLZXSZJPZZY25}
      & \xmark & \xmark & \cmark$^{a}$ & \ref{subsec:family-profile} \\
    CaMeL~\cite{Debenedetti2025CaMeL}
      & \xmark & \xmark & \cmark$^{b}$
      & \ref{subsec:evidence-overview}, \ref{subsec:continuation-results} \\
    Fides~\cite{costa2025fides}, formal configuration: integrity
      & \cmark$^{c}$ & \cmark & \cmark$^{c}$
      & \ref{subsec:upper-endpoint-results} \\
    Fides~\cite{costa2025fides}, formal configuration: confidentiality (explicit secrecy)
      & \cmark$^{c}$ & \cmark & \cmark$^{c}$
      & \ref{subsec:upper-endpoint-results} \\
    \bottomrule
  \end{tabular*}
  \par\vspace{3pt}
  \parbox{\textwidth}{\raggedright \cmark: counted in that column of
  Table~\ref{tab:evidence-profile}; \xmark: not counted.  In the full coverage
  column, each check mark is an argument.  $^{a}$Measured: 0\% accuracy without the TEE's
  authorization.  $^{b}$NORMAL mode: one private bit can leak through an
  exception.  $^{c}$Guaranteed zero.}
\end{table*}

\subsection{Screening and Coding}

\paragraph{Inclusion}
We include a paper if it concerns an LLM-based, LLM-integrated, or
LLM-agentic system and either asserts or tests whether an intervention changes an
adverse deployment outcome, shows that such an assertion fails, or supplies
instruments for such tests (e.g., benchmarks and evaluation-validity
critiques).  An attack paper is included because it tests whether current
deployments resist its attack class.

\paragraph{Safeguard evaluations}
During full-text analysis, a test that meets the definition in
Section~\ref{subsec:claim-instance-anchor} is coded as a safeguard
evaluation when the source identifies the safeguard configuration and states
how it measures harmful assistance.
Table~\ref{tab:app-funnel} gives the counts at each stage, and
Table~\ref{tab:evidence-list} lists the evaluations counted in the last three
columns of Table~\ref{tab:evidence-profile}.

\begin{table}[t]
\centering
\footnotesize
  \caption{Identification, screening, and coding counts.}
\label{tab:app-funnel}
\setlength{\tabcolsep}{1pt}
\begin{tabular*}{\columnwidth}{@{\extracolsep{\fill}}lr@{}}
\toprule
Stage & Count \\
\midrule
\multicolumn{2}{@{}l}{\emph{Identification}} \\
arXiv, scoped & 10{,}387 \\
ACL Anthology, scoped & 3{,}899 \\
dblp, 11 venues (2023--2026), title prefilter & 4{,}990 \\
\textbf{Records before the venue or citation filter} & \textbf{19{,}276} \\
\midrule
\multicolumn{2}{@{}l}{\emph{Screening}} \\
Passing the venue or citation filter & 5{,}804 \\
After version-family dedup & 4{,}810 \\
\textbf{Included by both screening models} & \textbf{1{,}906} \\
\midrule
\multicolumn{2}{@{}l}{\emph{Full-text analysis}} \\
Papers read until the shares converged & 197 \\
Papers with safeguard evaluations & 52 \\
Safeguard evaluations & 85 \\
\bottomrule
\end{tabular*}
\end{table}

\section{Worked Coding Example}
\label{app:full-results}
\label{subsec:case-fides}

Table~\ref{tab:case-fides-app} shows the coding of the integrity noninterference result
for Fides's formal configuration~\cite{costa2025fides}.

\begin{table*}[t]
  \centering
  \footnotesize
  \caption{Coding of Fides's formal configuration.}
  \label{tab:case-fides-app}
  \renewcommand{\arraystretch}{1.2}
  \setlength{\tabcolsep}{4pt}
  \setlength{\aboverulesep}{0pt}
  \setlength{\belowrulesep}{0pt}
  \setlength{\extrarowheight}{1pt}
  \rowcolors{2}{rowgray}{white}
  \begin{tabular}{>{\raggedright\arraybackslash}p{100pt}
    >{\raggedright\arraybackslash}p{345pt}}
    \toprule
    Element or field & Coding \\
    \midrule
    \multicolumn{2}{l}{\emph{Evaluation setting}} \\
    Service ($S$) & the source's deterministic agent loop and fixed toolset,
      before information-flow enforcement \\
    Intervention ($D$) & the enforcing agent loop with the taint-tracking
      planner and the source's information-flow policies, correctly applied \\
    Task, environment, and horizon ($\mathcal E$) & the source's labeled
      datastore and deterministic execution semantics, over a fixed finite
      prefix \\
    Outcome scale ($Z$) & binary failure of trusted-state noninterference
      under the source's execution comparison, which compares only executions
      with the same number of steps and the same terminal command; zero for
      every allowed substitution (the source's Proposition 1) \\
    Threat model ($\Sigma$) & all untrusted initial-datastore substitutions
      preserving the initial command and initial trusted state \\
    Score ($v$) & whether the trusted-state observations in the labeled
      execution prefix from the substituted datastore differ from those in the
      execution from the original datastore; the deterministic semantics fix
      that execution \\
    \midrule
    \rowcolor{white}
    \multicolumn{2}{l}{\emph{Fields}} \\
    Category & system and agent architecture \\
    Attack succeeds & no \\
    Miss-rate evidence & guaranteed ($\epsilon=0$); a required check allows a
      write to trusted state only when taint tracking labels every input to it
      trusted \\
    Coverage evidence & argued full ($\alpha=1$); every executed write to
      trusted state in the prefix passes its required call and argument
      checks \\
    Evidence after a block & yes ($r=0$); no reachable execution in the
      declared prefix fails the source's trusted-state comparison \\
    \bottomrule
  \end{tabular}
\end{table*}

\makeatletter\@dblfptop 0pt\relax\makeatother

\end{document}